\documentclass[11pt]{article}

\usepackage[T1]{fontenc}

\usepackage[USenglish]{babel}

\usepackage{microtype}

\usepackage[letterpaper, margin = 1in]{geometry}

\usepackage{mathtools}
\usepackage{amssymb}

\usepackage{amsthm}
\usepackage{thmtools}
\usepackage{thm-restate}

\usepackage{cleveref}

\usepackage{graphicx}

\usepackage{pgf,tikz}
\usepackage{mathrsfs}
\usetikzlibrary{arrows}

\theoremstyle{plain}
\newtheorem{theorem}{Theorem}[section]
\newtheorem{lemma}[theorem]{Lemma}

\newtheorem*{theorem*}{Theorem}
\newtheorem*{lemma*}{Lemma}

\theoremstyle{definition}
\newtheorem{definition}[theorem]{Definition}

\theoremstyle{remark}

\newtheorem{remark}[theorem]{Remark}

\newcommand{\es}{\emptyset}

\newcommand{\ceil}[1]{\left\lceil #1 \right\rceil}
\newcommand{\br}[1]{\left( #1 \right)}
\newcommand{\n}[1]{\lVert #1 \rVert}

\newcommand{\en}{\enspace}

\def\eps{ {\varepsilon} }
\def\N{ {\mathbb{N}} }

\def\R{ {\mathbb{R}} }
\def\C{ {\mathbb{C}} }

\def\E{ {\mathbb{E}} }
\def\O{ {\Omega} }

\DeclareMathOperator*{\argmin}{argmin}

\def\A{ {\mathcal{A}} }

\def\B{ {\mathbb{B}} }
\def\C{ {\mathcal{C}} }

\def\a{ {\alpha} }
\def\b{ {\beta} }
\def\g{ {\gamma} }
\def\c{ {\texttt{cost}} }
\def\o{ {\texttt{opt}} }
\def\d{ {\texttt{dim}} }
\def\l{ {\texttt{Lip}} }

\usepackage[square, authoryear]{natbib}
\usepackage{csquotes}

\title{The Role of Dimension in the Online Chasing Problem}
\author{Hristo Papazov}
\date{}

\begin{document}

\maketitle

\begin{abstract}
    Let $(X, d)$ be a metric space and $\C \subseteq 2^X$ -- a collection of special objects. In the $(X,d,\C)$-chasing problem, an online player receives a sequence of online requests $\{B_t\}_{t=1}^T \subseteq \C$ and responds with a trajectory $\{x_t\}_{t=1}^T$ such that $x_t \in B_t$. This response incurs a movement cost $\sum_{t=1}^T d(x_t, x_{t-1})$, and the online player strives to minimize the competitive ratio -- the worst case ratio over all input sequences between the online movement cost and the optimal movement cost in hindsight. Under this setup, we call the $(X,d,\C)$-chasing problem \textit{chaseable} if there exists an online algorithm with finite competitive ratio. In the case of Convex Body Chasing (CBC) over real normed vector spaces, \cite{BLLS} proved the chaseability of the problem. Furthermore, in the vector space setting, the dimension of the ambient space appears to be the factor controlling the size of the competitive ratio. Indeed, recently, \cite{sellke} provided a $d-$competitive online algorithm over arbitrary real normed vector spaces $(\R^d, \n{\cdot})$, and we will shortly present a general strategy for obtaining novel lower bounds of the form $\O(d^c), \en c > 0$, for CBC in the same setting.
    
    In this paper, we also ask the more general question of whether metric-space analogues of dimension exert a similar control on the hardness of chasing balls over metric spaces. We answer this question in the negative for the \textit{doubling} and the \textit{Assouad} dimensions of a metric space $(X, d)$. More specifically, we prove that the competitive ratio of chasing a nested sequence of balls can be arbitrary large for metric spaces of any doubling and Assouad dimensions. We further prove that for any large enough $\rho \in \R$, there exists a metric space $(X,d)$ of doubling dimension $\Theta(\rho)$ and Assouad dimension $\rho$ such that no online selector can achieve a finite competitive ratio in the general ball chasing regime.
\end{abstract}

\section{Introduction and Related Work}

In the Convex Body Chasing (CBC) problem for a specified real normed vector space $(\R^d, \n{\cdot})$, an online player/algorithm $\A$ begins at $x_0 = 0$, receives a convex body $K_t \subseteq \R^d$ as a request at turn $t \in \N$, and moves to a new point $x_t = \A_t(K_1, \dots, K_t) \in K_t$, incurring a movement cost $\n{x_t - x_{t-1}}$. After a finite sequence of convex requests $K_1, \dots, K_T$, the game ends with the player accumulating a total cost of
\begin{equation*}
    \c_T(x_1, \dots, x_T) = \sum_{t = 1}^T \n{x_t - x_{t-1}}.
\end{equation*}
Let us denote with
\begin{equation*}
    \o_T(K_1, \dots, K_T) = \min_{\{y_t \in K_t\}_{t=1}^T} \sum_{t = 1}^T \n{y_t - y_{t-1}}
\end{equation*}
the optimal movement cost which a clairvoyant algorithm could achieve. With these definitions, the competitive ratio of $\A$ becomes
\small
\begin{equation*}
    C(\A) = \sup \{\frac{\c_T(x_1, \dots, x_T)}{\o_T(K_1, \dots, K_T)} : T \in \N, \en K_t \subseteq \R^d - \text{ convex body}, \en x_t = \A_t(K_1, \dots, K_t), \en \forall t \in [T]\},
\end{equation*}
\normalsize
and the algorithm $\A$ is said to be $C(\A)-$competitive, where $C(\A) \in [1, \infty]$. Moreover, we define the CBC problem over $(\R^d, \n{\cdot})$ as \textit{chaseable} if there exists an online player with a finite competitive ratio in that setting. 

Notice that we tacitly assumed that $\A$ follows a deterministic strategy. Instead, we could have allowed the selection $x_t(w_t) = \A_t(K_1, \dots, K_t, w_t)$ to also depend on some random bits $w_t$. However, since $\c_T : (\R^d)^T \to \R_{\geq 0}$ is convex due to the convexity of the norm $\n{\cdot}$, Jensen's Inequality implies that the deterministic algorithm $\E[\A]$ defined through $\E[\A]_t = \E_{w_t} \A_t$ achieves a lower total cost for any sequence of convex requests $K_1, \dots, K_T$:
\begin{equation*}
    \c_T(\E_{w_1, \dots, w_T} \left[(x_1(w_1), \dots, x_T(w_T))\right]) \leq \E_{w_1, \dots, w_T} \left[\c_T(x_1(w_1), \dots, x_T(w_T))\right].
\end{equation*}

Hence, for the CBC problem, we can restrict our attention only to deterministic online algorithms. In a more general setting of the chasing problem, however, randomness becomes advantageous. Indeed, if we only keep the minimal structural assumptions necessary for the chasing problem to make sense, we arrive at the Metrical Service Systems (MSS) problem, where the online player moves inside a metric space $(X, d)$ and the online sequence of requests $S_1, \dots, S_T$ consists of arbitrary subsets of $X$.\footnote{Note that in the metric space setting, the performance of an online algorithm also depends on the starting point $x_0$, so the competitive ratio in this case will be calculated as a supremum over all allowed request sequences and starting configurations.} In this most abstract case of the chasing problem, \citet{MMS} proved $|X| - 1$ to be the best possible competitive ratio. Nonetheless, when the online player has access to random bits for point selection, the papers \citep{bubeck2021metrical} and \citep{bubeck2022randomized} establish matching upper and lower bounds of $\Theta(\log^2|X|)$ for the best randomized competitive ratio.

Therefore, as the size of the ambient metric space $X$ increases, the competitive ratio of the MMS problem grows prohibitively large. And clearly, when $|X| = \infty$, no finitely competitive deterministic or randomized online algorithm for MMS exists. However, with the introduction of some restrictions on the family of chased-after objects $\C \subseteq 2^X$, the problem setting becomes richer,\footnote{We invite the reader to consider the introduction of \citep{sellke} for references to many online problems -- such as the $k-$server problem, online set covering, and scheduling -- which emerge from metrical service systems through restrictions on the family of possible requests.} and competitive algorithms begin to appear even over infinite metric spaces.

The problem of convex body chasing over real normed vector spaces $(\R^d, \n{\cdot})$, first introduced by \citet{FL}, emerged as one such convex relaxation of the MMS problem. Friedman and Linial proved a finite competitive ratio in the two-dimensional Euclidean setting and conjectured that a finitely competitive algorithm exists in general. The initial efforts towards resolving the conjecture proved the chaseability of the problem for subspace chasing \citep{ABNPSS} and \textit{nested}\footnote{The adversary can request only collapsing sequences of convex bodies $K_1 \supseteq \dots \supseteq K_T$.} convex body chasing \citep{BBEKU}. Later work from Argue et al. \citet{ABCGL} gave an $O(d \log d)-$competitive algorithm for nested CBC over any norm, and most recently in the same regime, \citet{BKLLS} provided an algorithm, which moves to a mixture of the Steiner point and the centroid of a convex body and achieves a competitive ratio optimal up to an $O(\log d)$ factor over $\ell_p-$normed spaces. In the general (non-nested) setting, \citet{BLLS} resolved the conjecture for Euclidean spaces by demonstrating the existence of a $2^{O(d)}-$competitive online algorithm. Finally, \citet{sellke} and \citet{AGTG} independently gave an $O(\min(d, \sqrt{d \log T}))-$competitive algorithm for convex body chasing in Euclidean spaces. Furthermore, Sellke's Steiner-point-based algorithm achieves competitive ratio at most $d$ for arbitrary normed spaces, which matches the lower bound for the $\ell_\infty$ norm.

We stated above that the nested CBC algorithm of \citet{BKLLS} performs nearly optimally in $\ell_p$ spaces, and we described the algorithm from \citep{sellke} as the optimal for $\ell_\infty$ spaces. We will now discuss existing lower bounds on the competitive ratio for nested CBC in order to justify these claims.

\begin{definition}
    Let $\mathcal{D}(X, d, \C)$ be the set of all all deterministic online chasing algorithms in the $(X, d, \C)$ setting. We define the competitive ratio of the $(X, d, \C)$-chasing problem as the following infimum:
    \begin{equation*}
        R(X, d, \C) = \inf_{\A \in \mathcal{D}(X, d, \C)} C(\A).
    \end{equation*}
    We similarly define $\hat{R}(X, d, \C)$ as the competitive ratio of the nested $(X, d, \C)$-chasing problem.
\end{definition}

Note that we will only write $R(X)$ and $\hat{R}(X)$ when the distance metric $d$ and the family of chased-after sets $\C$ are clear from context. Also, notice that trivially $R(X) \geq \hat{R}(X)$.

Now, by sequentially considering the faces of the unit hypercube $$F_0 = \{\mathbf{z} \in \R^d : \n{z}_\infty = 1\}, \en F_t = F_{t-1} \cap \{\mathbf{z} \in \R^d : z_t = \pm 1\}, \en t \geq 1,$$ depending on which one stands farther away from the online selector at time $t$, \citet{FL} showed a lower bound of $d$ on the competitive ratio for the nested $\ell_\infty$ case. The same construction leads to a lower bound of $d^{1 - \frac{1}{p}}$ for $d-$dimensional $\ell_p$ spaces, $p \geq 1$. Apart from these unit-cube lower bounds, \citet{BKLLS} used the existence of a Hadamard matrix in dimensions $d = 2^k$ to prove that $\hat{R}(\ell_p^d) \geq \O(\sqrt{d})$, where $\ell_p^d = (\R^d, \n{\cdot}_p)$. Thus, the following lower bounds on the nested CBC competitive ratio have been proven to hold:
\begin{itemize}
    \item $\hat{R}(\ell_p^d \geq \O(d^{1 - \frac{1}{p}})$ for $d \in \N$, $p \in [0, \infty]$;
    \item $\hat{R}(\ell_p^d \geq \O(\max(d^{1 - \frac{1}{p}}, \sqrt{d}))$ for $d = 2^k, \en k \in \N$, $p \in [0, \infty]$.
\end{itemize}
In particular, no lower bounds have been established for non-$\ell_p$ norms.

\paragraph{Our Results.}

As our first main result, we develop a general strategy for transferring bounds on the competitive ratio of the chasing problem from one setting $(X, d_X, \C_X)$ to another $(Y, d_Y, \C_Y)$. We then connect this general approach with several results from Convex Geometry to prove the following theorem, which generalizes the aforementioned lower bound from \citet{BKLLS} and provides the first known lower bound on nested CBC for arbitrary normed spaces:
\begin{theorem} \label{lower_bounds}
    For every $d \in \N$ and $p \in [0, \infty]$, no online algorithm can achieve a competitive ratio better than $\O(\max(d^{1 - \frac{1}{p}}, \sqrt{d}))$ in the nested convex body chasing problem over $\ell_p^d$. Furthermore, no online algorithm can achieve a competitive ratio better than $\O(d^{1/6})$ in the nested convex body chasing problem over arbitrary normed spaces $(\R^d, \n{\cdot})$.
\end{theorem}

Our result combined with Sellke's algorithm \citep{sellke} proves that the competitive ratio of CBC satisfies the inclusion $R(\R^d, \n{\cdot}) \in \left[ \O(d^{1/6}), d \right]$ for any norm $\n{\cdot}$. Hence, the vector-space dimension $d$ appears to be the quantity controlling the hardness of CBC.

As our second research direction, we investigate the question of whether metric-space analogues of dimension exert a similar control over the competitive ratio $R(X, d, \C)$ in the general $(X, d, \C)$-chasing problem. Of course, the family of chased-after objects $\C$ can no longer be the set of convex bodies since convexity requires a linear structure. We also want to impose some meaningful restrictions on $\C$ in order for the problem be interesting because setting $\C = 2^X$ brings us back to the already resolved metrical service systems problem. Hence, we decide to consider closed balls -- an object native to any metric space -- as a metric-space substitute for convex-body requests.
\begin{definition}
    Given $x \in X$ and $R \geq 0$, the closed ball $\bar{B}(x, R)$ is defined as the set of points within radius $R$ from $x$:
    \begin{equation*}
        \bar{B}(x, R) = \{ y \in X : d(x, y) \leq R \}.
    \end{equation*}
    We will use $B(x, R)$ to denote the open ball of radius $R$ around $x$.
\end{definition}
We can now set $\C = \{ \bar{B}(x, R) : x \in X, \en R \geq 0 \}$ to fully specify the Ball Chasing (BC) problem and the Nested Ball Chasing (NBC) problem with competitive ratios $R(X) \geq \hat{R}(X)$, which we will try to characterize in terms of the $\g-$cover dimension and the Assouad dimension of $X$. We will postpone the definitions of the $\g-$cover and Assouad dimensions to Section \ref{preliminaries} and proceed with stating our NBC hardness results.

\begin{restatable}{theorem}{doubling} \label{doubling}
    For any $\g > 1$ and $\rho, N > 0$, there exists a metric space $(X, d)$ with $\g-$cover dimension $\d_\g(X) = \Theta_\g(\rho)$ such that no deterministic online algorithm can achieve a competitive ratio better than $N$ in the nested ball chasing problem over $X$. In other words, $R(X) \geq \hat{R}(X) \geq N$.
\end{restatable}

\begin{remark}
    We use the notation $\Theta_\g(\cdot), \en \O_\g(\cdot)$, and $O_\g(\cdot)$ to indicate that the hidden constants depend only on $\g$.
\end{remark}

\begin{restatable}{theorem}{assouad} \label{assouad}
    There exists an absolute constant $c > 0$ such that for every $\rho \geq c$, there exists a countable metric space $(X, d)$ with Assouad dimension $\d_A(X) = \rho$ such that no deterministic online algorithm can achieve a finite competitive ratio in the nested ball chasing problem over $X$. In other words, $R(X) = \hat{R}(X) = \infty$.
\end{restatable}

\section{Competitive Ratio Bounds through Metric Embeddings}

Our goal in this section will be to prove Theorem \ref{lower_bounds} by developing a general method for transferring competitive-ratio information between different settings of the chasing problem. Without further ado, let $(X, d_X, \C_X)$ and $(Y, d_Y, \C_Y)$ be two distinct chasing environments.

\begin{lemma} \label{transfer}
    Suppose there exists a bi-Lipschitz metric embedding $f : (X, d_X) \to (Y, d_Y)$ with distortion $D(f) = \l(f)\l(f^{-1})$such that $\forall K \in \C_X, \en f(K) \in \C_Y$. Then, every deterministic online algorithm $\A_Y \in \mathcal{D}(Y, d_Y, \C_Y)$ with competitive ratio $C(A_Y)$ induces a deterministic online algorithm $\A_X \in \mathcal{D}(X, d_X, \C_X)$ through $f^{-1}$ such that $C(A_X) \leq C(A_Y)D(f)$. 
\end{lemma}

\begin{proof}
    We define the deterministic online algorithm $\A_X \in \mathcal{D}(X, d_X, \C_X)$ the following way. Let $y_0 = f(x_0)$ be the starting point in the $Y-$space. Given a sequence of online requests $\{K_t\}_{t=1}^T \subseteq \C_X$, we run $\A_Y$ on the embedded request sequence $\{f(K_t)\}_{t=1}^T \subseteq \C_Y$ to obtain an online trajectory $\{y_t\}_{t=0}^T$. We then define the action of $\A_X$ at time $t$ as the movement towards point $x_t = f^{-1}(y_t)$. Now,
    \begin{equation*}
        \c_T(x_0, x_1, \dots, x_T) = \sum_{t=1}^T d(x_t, x_{t-1}) \leq \l(f^{-1}) \sum_{t=1}^T d(y_t, y_{t-1}) = \l(f^{-1}) \c_T(y_0, y1, \dots, y_T).
    \end{equation*}
    Moreover,
    \begin{align*}
        \o_T(x_0, K_1, \dots, K_T) & = \inf_{\{a_t \in K_t\}_{t=1}^T} \sum_{t = 1}^T \n{a_t - a_{t-1}} \\
        & \geq \inf_{\{a_t \in K_t\}_{t=1}^T} \sum_{t = 1}^T \frac{1}{\l(f)}\n{f(a_t) - f(a_{t-1})} \\
        & \geq \frac{1}{\l(f)} \inf_{\{b_t \in f(K_t)\}_{t=1}^T} \sum_{t = 1}^T \n{b_t - b_{t-1}} \\
        & = \frac{1}{\l(f)} \o_T(f(x_0), f(K_1), \dots, f(K_T)).
    \end{align*}
    Hence,
    \begin{align*}
        C(\A_X) & = \sup_{x_0 \in X, T \in \N, \{K_t\}_{t=1}^T \subseteq \C_X} \frac{\c_T(x_0, x_1, \dots, x_T)}{\o_T(x_0, K_1, \dots, K_T)} \\
        & \leq \sup_{x_0 \in X, T \in \N, \{K_t\}_{t=1}^T \subseteq \C_X} \l(f)\l(f^{-1}) \frac{\c_T(y_0, y_1, \dots, y_T)}{\o_T(f(x_0), f(K_1), \dots, f(K_T))} \\
        & \leq \sup_{y_0 \in Y, T \in \N, \{S_t\}_{t=1}^T \subseteq \C_Y} D(f) \frac{\c_T(y_0, y_1, \dots, y_T)}{\o_T(y_0, S_1, \dots, S_T)} \\
        & = D(f) C(\A_Y).
    \end{align*}
\end{proof}

Now, let $f : (X, d_X) \to (Y, d_Y)$ be a bi-Lipschitz metric embedding which maps $\C_X$ to $\C_Y$. For any algorithm $\A_Y \in \mathcal{D}(Y, d_Y, \C_Y)$, let us denote with $f^{-1}(\A_Y) \in \mathcal{D}(X, d, \C)$ the online algorithm obtained by the procedure in the above lemma. Then,
\begin{align*}
    R(X) & = \inf_{\A_X \in \mathcal{D}(X, d, \C)} C(\A_X) \\
    & \leq \inf_{f^{-1}(\A_Y) : \A_Y \in \mathcal{D}(Y, d_Y, \C_Y)} C(f^{-1}(\A_Y)) \\
    & \leq \inf_{\A_Y \in \mathcal{D}(Y, d_Y, \C_Y)} D(f) C(A_Y) \\
    & = D(f) R(Y).
\end{align*}
Similarly, $\hat{R}(X) \leq D(f) \hat{R}(Y)$. Furthermore, if the metric embedding $f$ is bijective such that $f(\C_X) = \C_Y$, then
\begin{equation} \label{metric_distortion}
    \boxed{\frac{1}{D(f)} \leq \frac{R(X)}{R(Y)} \leq D(f) \en\en\en \text{and} \en\en\en
    \frac{1}{D(f)} \leq \frac{\hat{R}(X)}{\hat{R}(Y)} \leq D(f).}
\end{equation}

Thus, metric embeddings which also preserve chasing families allow us to transfer competitive ratio upper and lower bounds between different problem settings. Let us see how we can apply this approach to derive novel lower bounds on the competitive ratio for convex body chasing over different norms.

When $(X, \n{\cdot}_X)$ and $(Y, \n{\cdot}_Y)$ are $d-$dimensional real normed vector spaces and the chased-after objects are convex sets, we will focus on non-singular linear maps $T : X \to Y$ as our choice for metric embeddings since linear maps conserve convexity under the image and preimage operations. We will denote with $\n{T} = \sup_{x : \n{x}_X = 1} \n{Tx}_Y$ and $\n{T^{-1}} = \sup_{y : \n{y}_Y = 1} \n{T^{-1}y}_X$ the operator norms of the bijective linear maps $T$ and $T^{-1}$. Hence, $D(T) = \n{T}\n{T^{-1}}$. Now, the Banach-Mazur distance\footnote{We refer the reader to \citep{tomczak-jaegermann} for a thorough introduction to the geometry of Banach spaces and the properties of the Banach-Mazur distance.} between the $d-$dimensional Banach spaces $X$ and $Y$ is defined as
\begin{equation*}
    d_{BM}(X, Y) = \inf \{D(T) | T : X \to Y - \text{isomorphism} \},
\end{equation*}
and we conclude that
\begin{equation} \label{vector_distortion}
    \boxed{\frac{1}{d_{BM}(X, Y)} \leq \frac{R(X)}{R(Y)} \leq d_{BM}(X, Y) \en\en\en \text{and} \en\en\en
    \frac{1}{d_{BM}(X, Y)} \leq \frac{\hat{R}(X)}{\hat{R}(Y)} \leq d_{BM}(X, Y).}
\end{equation}

The unit-cube construction from \citep{FL} gave us the lower bound $\hat{R}(\ell_p^d) \geq d^{1 - \frac{1}{p}}$, where $p \geq 1$. Hence, for any $d-$dimensional Banach space $X$, inequality (\ref{vector_distortion}) implies
\begin{equation*}
    \hat{R}(X) \geq \frac{d^{1 - \frac{1}{p}}}{d_{BM} (X, \ell_p^d)}.
\end{equation*}
Let $\mathcal{B}_d$ be the set of all $d-$dimensional Banach spaces. As a starting point of our search for convenient Banach-Mazur distances to $\ell_p$ spaces, we recall John's Theorem \citep{john} (see \citep{ball} for proof and discussion), which states that $\sup_{X \in \mathcal{B}_d}d_{BM}(X, \ell_2^d) = \sqrt{d}$. Unfortunately, this bound leads to the trivial inequality $\hat{R}(X) \geq 1$. However, bounds on the Banach-Mazur distance to $\ell_\infty^d$ prove more fruitful.\footnote{\citet{tikhomirov} proved the best known lower bound $\inf_{X \in \mathcal{B}_d}d_{BM}(X, \ell_\infty^d) \geq d^{5/9} \log^{-C} d$ for some universal constant $C$.} Indeed, \citet{giannopoulos} gave the asymptotic upper bound $\sup_{X \in \mathcal{B}_d}d_{BM}(X, \ell_\infty^d) \leq O(d^{5/6})$, which was later improved to the precise estimate $\sup_{X \in \mathcal{B}_d}d_{BM}(X, \ell_\infty^d) \leq 2d^{5/6}$ by \citet{youssef}. Thus, we arrive at the nontrivial lower bound on the competitive ratio of CBC over arbitrary Banach spaces
\begin{equation*}
    R(X) \geq \hat{R}(X) \geq \frac{1}{2} d^{1/6}, \en \forall X \in \mathcal{B}_d.
\end{equation*}

Furthermore, when dealing with $\ell_p$ norms, Propostion 37.6 from \citep{tomczak-jaegermann} states that for any dimension $d \in \N$ and $1 \leq p < 2 < q \leq \infty$, $d_{BM}(\ell_p^d, \ell_q^d) = \Theta(d^\a)$, where $\a = \max(\frac{1}{p} - \frac{1}{2}, \frac{1}{2} - \frac{1}{q})$. Therefore, setting $q = \infty$, we obtain $d_{BM} = \Theta(\sqrt{d})$. Thus, inequality (\ref{vector_distortion}) entails
\begin{equation*}
    R(\ell_p^d) \geq \hat{R}(\ell_p^d) \geq \O(\sqrt{d}), \en \forall p \in [1,2),
\end{equation*}
which concludes the proof of Theorem \ref{lower_bounds}.

\section{Metric-Dimension Preliminaries} \label{preliminaries}

Having proved that the vector-space dimension controls the hardness of CBC, we proceed to investigate the effect of metric covering dimensions on the competitive ratio of the ball chasing problem. However, before we jump into the proofs of our main theorems, we first need to set up some notation and preliminary observations.

Given a metric space $(X, d)$ and a subset $E \subseteq X$, we denote with $N_r(E)$ the smallest number of open balls in $X$ with radius at most $r$ that cover $E$.
\begin{definition}
For $\g > 1$, the \textbf{$\g-$cover constant} of a metric space $(X, d)$ is the smallest integer $\lambda_\g$ such that for every $R > 0$ and every $x \in X$, the open ball $B(x, R)$ can be covered by at most $\lambda_\g$ open balls of radius $R/\g$. In other words,
\begin{equation*}
    \lambda_\g = \sup_{x \in X, R > 0} N_{R/\g}(B(x, R)).
\end{equation*}
The \textbf{$\g-$cover dimension} of $X$ is then defined as $\rho_\g = \log_\g \lambda_\g$ and denoted as $\d_\g(X)$. Whenever $\d_\g(X) < \infty$, the metric space $X$ is said to be \textbf{doubling}.
\end{definition}

The above definition seemingly generalizes the doubling dimension ($\g = 2$) introduced in \citep{GKL}. However, the $\g-$cover dimension still suffers from the inherent arbitrariness of the choice of $\g$ and should be considered as a precise quantity only up to a constant factor. Indeed, notice that for $1 < \a < \b$,
\begin{equation*}
    \lambda_\a \leq \lambda_\b \leq \lambda_\a^{\ceil{\log_\a \b}},
\end{equation*}
so
\begin{equation*}
    \d_\a(X) \leq \d_\b(X) < (1 + \log_\b \a)\d_\a(X).
\end{equation*}

Intuitively, the $\d_\g(X)$ expresses the extent to which the size of the space increases as you travel further away and thus captures a relative measure of the dimension of the space. We make this intuition concrete through \Cref{doubling_vector} which relates the $\g-$cover dimension of $(R^d, \n{\cdot})$ to the vector-space dimension $d$. The proofs of all lemmas in this preliminaries section can be found in Appendix \ref{proofs}.\footnote{We make no claims about the novelty of the preliminary lemmas, and we note that \Cref{doubling_vector}, \Cref{doubling_assouad}, and \Cref{assouad_vector} can be found in \citep{heinonen} and \citep{matousek}.}

\begin{restatable}{lemma}{doublingVector} \label{doubling_vector}
Let $\g > 1$ and let $(X, d) = (\R^d, \n{\cdot})$. Then, $\d_\g(X) = \Theta_\g(d)$.
\end{restatable}

Next, we demonstrate the robustness of the $\g-$cover dimension under the submetric relation -- a property shared with the ordinary vector-space dimension.

\begin{restatable}{lemma}{doublingSubmetric} \label{doubling_submetric}
Let $\g > 1$ and let $(Y,d|_Y)$ be a metric subspace of $(X,d)$. Then, $\d_\g(Y) \leq O_\g(\d_\g(X))$.
\end{restatable}

As a final property of the $\g-$cover dimension that will be useful in our proofs, we show an upper bound on the $\g-$cover dimension of the union of two distant metric subspaces.

\begin{restatable}{lemma}{doublingUnion} \label{doubling_union}
Let $\g > 1$ and let $(Z, d)$ be a metric space such that $Z = X \cup Y$. If $d(X, Y) \geq 2\g \max\{\texttt{diam}(X), \texttt{diam}(Y)\}$, then $\d_\g(Z) = \max\{\d_\g(X), \d_\g(Y)\}$.
\end{restatable}

Now, we introduce a more precise metric dimension called the Assouad dimension \citep{assouad} or the (metric) covering dimension.\footnote{We invite the reader to consider Chapter 10 of \citep{heinonen}, Chapter 3 of \citep{matousek}, and \citep{fraser} for an in-depth discussion of doubling metric spaces and the properties of the Assouad dimension.}
\begin{definition}
    The \textbf{Assouad dimension} of a metric space $(X, d)$ is defined as the infimal $\rho_A \geq 0$ for which there exists a constant $C > 0$ such that for all radii $0 < r < R$ and all $x \in X$,
    \begin{equation*}
        N_r(B(x, R)) \leq C \left( \frac{R}{r} \right)^{\rho_A}.
    \end{equation*}
    We will denote the Assouad dimension of $X$ with $\d_A(X)$.
\end{definition}

Notice that finite metric spaces have Assouad dimension 0. Also, not surprisingly:

\begin{restatable}{lemma}{doublingAssouad} \label{doubling_assouad}
    The metric spaces with finite Assouad dimension are precisely the doubling metric spaces. Moreover, whenever $(X, d)$ is doubling and $\g > 1$, $\d_A(X) \leq \d_\g(X)$.
\end{restatable}

Furthermore, the Assouad dimension perfectly agrees with the vector-space dimension.

\begin{restatable}{lemma}{assouadVector} \label{assouad_vector}
    Let $(X, d) = (\R^d, \n{\cdot})$. Then, $\d_A(X) = d$.
\end{restatable}

We conclude the preliminaries section with an intermediate-value theorem for the Assouad dimension that we will use in the proof of Theorem \ref{assouad}.

\begin{theorem} \label{ivp}
    Let $(X, d)$ be a doubling metric space.
    \begin{itemize}
        \item \citet{CWW} proved that for any $\rho \in [0, \d_A(X)]$, there exists a metric subspace $(Y, d|_Y) \subseteq (X, d)$ such that $\d_A(Y) = \rho$.
        \item \citet{WW} proved that when $X$ is bounded and $\rho \in [0, \d_A(X))$, then there exists a countable metric subspace $(Y, d|_Y) \subseteq (X, d)$ such that $\d_A(Y) = \rho$.
    \end{itemize}
\end{theorem}

\section{Arbitrary Bad Competitive Ratio in Nested Ball Chasing}

\doubling*

\begin{proof}
We are given as input $\rho > 0$ and $N > 0$, and we want to construct a metric space $(X,d)$ with $\d_\g(X) = \Theta_\g(\rho)$ such that $\hat{R}(X) \geq N$.

\paragraph{\textbf{Defining the Metric Space.}}

The metric space $(X,d)$ will consist of two parts: an $n-$point set $P = \{p_1, p_2, \dots, p_n\}$ of "nearby" points and a set of $2^n - 2$ "faraway" points $Z = \{z_S : S \subseteq [n], \en S \neq \es, [n]\}$ that the adversary will use for ball centers. The hard chasing instances will be designed in such a way that the online player starts at a point in $P$ and is forced to traverse the whole of $P$ before the nested request sequence terminates. The points in $Z$ will be placed so far away from $P$ that crossing over to them will immediately inflict a competitive ratio worse than $N$. Thus, the strategy of the online adversary will be to adaptively construct a sequence of nested balls $B_1 \supset \dots \supset B_{n-1}$ such that at time $t \geq 1$, $B_t \setminus B_{t-1} = x_{t-1}$ -- the point currently occupied by the online player. This way, the online player will necessarily travel a distance of at least $(|P|-1)d_{\min}(P)$, while the optimal path $d(x_0, x_n)$ would have length at most $d_{\max}(P)$. Hence, no deterministic online algorithm can achieve competitive ratio better than $(|P|-1)/\Delta(P)$, where $\Delta(P) = d_{\max}(P)/ d_{\min}(P)$ is the aspect ratio of the metric subspace $P$. Therefore, we now have the concrete goal of designing a metric space $(X, d)$ where the $\g-$cover dimension has no control over the size of $(|P|-1)/\Delta(P)$. We proceed with the description of $(X,d)$.

\begin{figure}[ht] \label{fig: finite}
    \centering
    \includegraphics[scale=0.6]{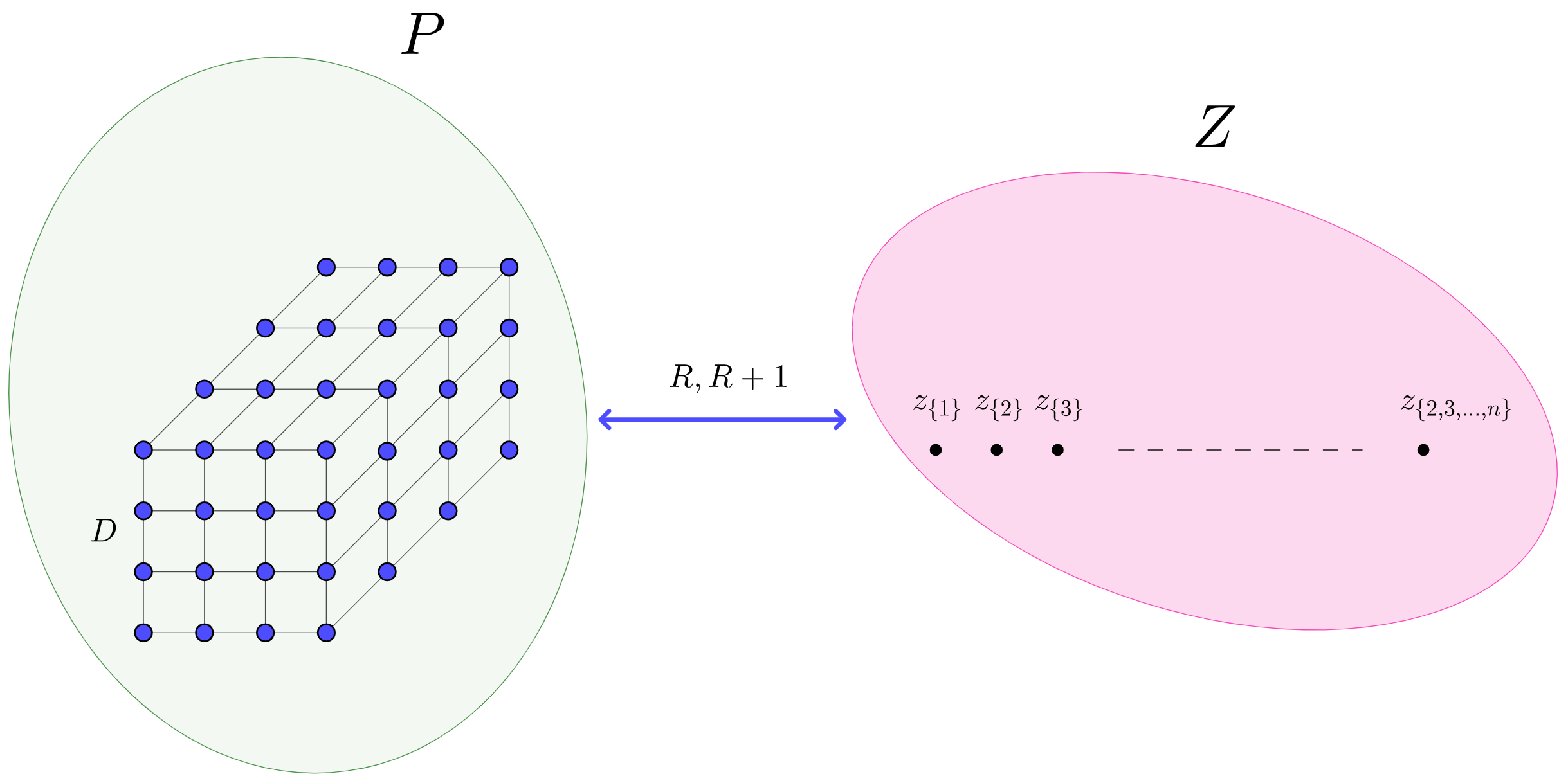}
    \caption{$(X,d)$ with $\d_\g(x) = \Theta_\g(\rho)$ and $\hat{R}(X) \geq N$.}
\end{figure}

Let $P \subseteq \R^k$ be a $k-$dimensional grid with unit distance between adjacent points and $D$ points on each side of the grid. Hence, $n = D^k$. The grid $P$ will be endowed with the $\ell_\infty$ metric, so $d(p_i,p_j) \in \{1,2, \dots, D-1\}, \en \forall i \neq j \in [n]$. Next, we specify the ball-centers set $Z$. For each $S \subseteq [n]: \en S \neq \es, [n]$, we place a point $z_S$ on a line (in arbitrary order) at unit distance from the neighboring points. Thus, the distance between the first and the last element on the ball-centers line becomes $2^n - 3$. Then, for each $j \in [n]$ and $S \subseteq [n]: \en S \neq \es, [n]$, we let
\begin{equation*}
    d(z_S, p_j) =
    \begin{cases}
        R + 1 & \text{if } j \in S;\\
        R & \text{otherwise},
    \end{cases}
\end{equation*}
where we set $R = \g 2^{n+1} > 2\g\max(\texttt{diam}(P), \texttt{diam}(Z))$.

\paragraph{\bf Proof of Valid Metric.}

Clearly, our metric, as defined, satisfies symmetry and positivity, so we only need to check that the triangle inequality holds. Since the triangle inequality holds separately within $P$ and $Z$, we just need to convince ourselves that no triangle with vertices both in $P$ and $Z$ violates the triangle inequality.

\paragraph{\bf Nested Sequence of Balls.}

Since $R > \texttt{diam}(Z)$, we get that $Z \subseteq B(z_S, R)$, and thus
\begin{equation*}
    B(z_S, R) = Z \cup \{x_j : j \notin S\}.
\end{equation*}
Now, let $S_t$ be the set of indices of points in $P$ that the online player has visited prior to the ball request at time $t \geq 1$. Then, the adaptive adversary can pick $B(z_{S_t}, R)$ and force the online player to move to another point in $P$ or $Z$. If the online player decides to jump to $Z$, the request sequence terminates. Otherwise, the adversary continues to request balls of the form $B(z_{S_{t+1}}, R) =  B(z_{S_t}, R) \setminus \{x_t\}$ until all points in $P$ have been exhausted.

\paragraph{\bf Bounding the Doubling Dimension.}

We proceed with the bounding of the $\g-$cover dimension of $X$. Clearly, the line $(Z, d|_Z)$ can be isometrically embedded into $(\R, |\cdot|)$, so by Lemma \ref{doubling_submetric}, $\d_\g(Z, d) \leq O_\g (\d_\g(\R, |\cdot|))$. Now, Lemma \ref{doubling_vector} tells us that $\d_\g(\R, |\cdot|) = \Theta_\g(1)$, so we conclude that $\d_\g(Z, d) \leq O_\g(1)$. Again, by Lemma \ref{doubling_vector}, $\d_\g(\R^k, \n{\cdot}_\infty) = \Theta_\g(k)$, so Lemma \ref{doubling_submetric} yields $\d_\g(P) \leq O_\g(\d_\g(k))$. Therefore, Lemma \ref{doubling_union} helps us establish that $\d_\g(X) \leq O_\g(k)$.

On the other hand, if we consider a point $x$ at the interior of the grid $P$, the ball $B(x, \frac{2\g-1}{\g})$ will contain $3^k$ points, but every ball of radius $\frac{2\g-1}{\g^2}$ will contain precisely a single point. Hence, $\lambda_\g(P) \geq 3^k$, which implies that $\d_\g(P) \geq \O_\g(k)$. Now, Lemma \ref{doubling_union} entails $\d_\g(X) \geq \O_\g(k)$, so overall $\d_\g(X) = \Theta_\g(k)$. Thus, we let $k = \ceil{\rho} + 1$, which concludes the proof that $\d_\g(X) = \Theta_\g(\rho)$.

\paragraph{\bf Bounding the Competitive Ratio.}

As a last step in the proof, we lower-bound the competitive ratio $\hat{R}(X)$. Note that due to the design of the nested sequence of requests, the optimal traveled distance never exceeds $d_{\max}(P) = D - 1$. Thus, if the online player decides to jump to $Z$ at some point in time, the competitive ratio of the chasing sequence becomes greater than $R/D >> (|P| - 1)/\Delta(P)$. Otherwise, as we discussed in the beginning of the section, the competitive ratio becomes at least $(|P| - 1)/\Delta(P) > D^{k-1} > D^\rho$. Now, since $D$ represents a free parameter, we can pick $D$ large enough so that $\hat{R}(X) > D^\rho \geq N$.

\end{proof}

\begin{remark}
Exactly the same construction shows that for any $N > 0$ and $\rho \geq 1$, there exists a metric space $(X,d)$ of KR-dimension\footnote{See \citep{KR} for the precise definition and origin of the KR-dimension.} $\Theta(\rho)$ such that $\hat{R}(X) \geq N$. 
\end{remark}

\section{Greedy Upper Bound on $\hat{R}(X)$ for Nested Ball Chasing}

In this section, we show a simple greedy algorithm for the NBC problem over arbitrary metric spaces $(X, d)$ (without doubling assumptions), which matches up to a constant factor the lower bound on $\hat{R}(X)$ that we proved in the doubling setting. Indeed, given a starting point $x_0$ and a nested request sequence $\mathcal{S} = \{B_1, \dots, B_T\}$, let us define $P_{x_0,\mathcal{S}} = B(x_0, d^*)$ as the \textbf{problem-space set}, where $d^* = d(x_0, B_T)$ stands for the optimal traveled distance. With this notation, the proof of Theorem \ref{doubling} demonstrated that there exists a metric space $(X, d)$ with fixed $\g-$cover dimension such that
\begin{equation*}
    \hat{R}(X) \geq \sup_{x_0 \in X, \mathcal{S} - \text{ nested ball sequence}} \frac{|P_{x_0,\mathcal{S}}|-1}{\Delta(P_{x_0,\mathcal{S}})}.
\end{equation*}
The theorem below shows once more that contrary to the situation in the vector-space setting, a fixed $\g-$cover dimension does not make the NBC problem easier over metric spaces.

\begin{theorem} \label{greedy}
    Let $(X, d)$ be an arbitrary metric space on which we consider the nested ball chasing problem. Then, with the above notation, $$\hat{R}(X) \leq \sup_{x_0 \in X, \mathcal{S} - \text{ nested ball sequence}} 2|P_{x_0,\mathcal{S}}|-1.$$
\end{theorem}

\begin{proof}
Let us consider a greedy online algorithm which starts at $x_0$ and responds with $x_t = \argmin_{x \in B_t} d(x,x_0)$ -- the closest point to $x_0$ in $B_t$. WLOG, assume $x_{t-1} \notin B_t$. (Otherwise, the algorithm simply skips the step without incurring any penalty.) Notice that $d^* = d(x_0, B_T) = d(x_0, x_T) \geq d(x_0, x_t)$ for $t \in [T]$. Hence, using the triangle inequality, we bound the online cost as follows:
\begin{equation*}
    d(x_0,x_1) + \dots + d(x_{T-1}, x_T) \leq d(x_0, x_T) + 2\sum_{t=1}^{T-1} d(x_0, x_t) \leq (2T-1) d^*.
\end{equation*}

Thus, the greedy algorithm achieves a competitive ratio of at most $2T-1$. Furthermore, we know that $T$ must be bounded by $|B(x_0, d^*)| = |P_{x_0,\mathcal{S}}|$ since the assumption $x_{t-1} \notin B_t$ guarantees that each step evicts an element of $B(x_0, d^*)$.

\end{proof}

\section{Infinite Competitive Ratio for Any Fixed Assouad Dimension}

We will now show that enforcing a finite Assouad dimension on a metric space $(X, d)$ proves to be too liberal of a constraint for the NBC problem to be chaseable. Hence, the question of finding a condition which allows for meaningful ball chasing in the absence of linear structure remains open.

\assouad*

\begin{proof}

We want to build a metric space $(X, d)$ with $\d_A(X) = \rho$ and $\hat{R}(X) = \infty$. The construction of that space will be carried out in two phases. The first phase will use induction to glue together countably many copies of the metric space we designed for Theorem \ref{doubling}. This gluing procedure will ensure that $\hat{R}(X) = \infty$. Then, in phase two, we will attach the result from phase one to a subset of a normed vector space of proper dimension, which will guarantee that $\d_A(X) = \rho$. Let us begin.

\paragraph{Phase I.}
Let us fix some $\g > 1$. For every $N \in \N$, Theorem \ref{doubling}, provides the existence of a metric space $(X_N, d_N)$ with NBC competitive ratio $\hat{R}(X_N) \geq N$ and $\g-$cover dimension $\d_\g(X_N) = \Theta_\g(1) = c$, where $c$ represents the absolute constant in the statement of Theorem \ref{assouad}. We set $(Y_1, D_1) = (X_1, d_1)$, and for $N \geq 2$, we construct the metric spaces $(Y_N, D_N)$ iteratively by gluing together $(X_N, d_N)$ and $(Y_{N-1}, D_{N-1})$ in the following way:
\begin{itemize}
    \item The distances between pairs of points both inside either $X_N$ or $Y_{N-1}$ remain unchanged.
    \item $\forall x \in X_N, \en \forall y \in Y_{N-1}$, we let $D_N(x, y) = R_N = 2\g \max(\texttt{diam}(X_N), \texttt{diam}(Y_{N-1}))$.
\end{itemize}
Clearly, the triangle inequality holds inside every so-constructed space $(Y_N, D_N)$. Moreover, a simple inductive argument in conjunction with Lemma \ref{doubling_union} yields $\d_\g(Y_N) = c$. Finally, since the competitive ratio $\hat{R}(\cdot)$ is monotonically increasing under the superset operation, we get obtain $\hat{R}(Y_N) \geq N$. Therefore, the following two identities hold for the metric space $(Y, D) = \bigcup_{N = 1}^\infty (Y_N, D_N)$:
\begin{itemize}
    \item $\d_\g(Y) = c$;
    \item $\hat{R}(Y) = \infty$.
\end{itemize}

\begin{figure}[h!] \label{fig: infinite}
    \centering
    \includegraphics[scale=0.45]{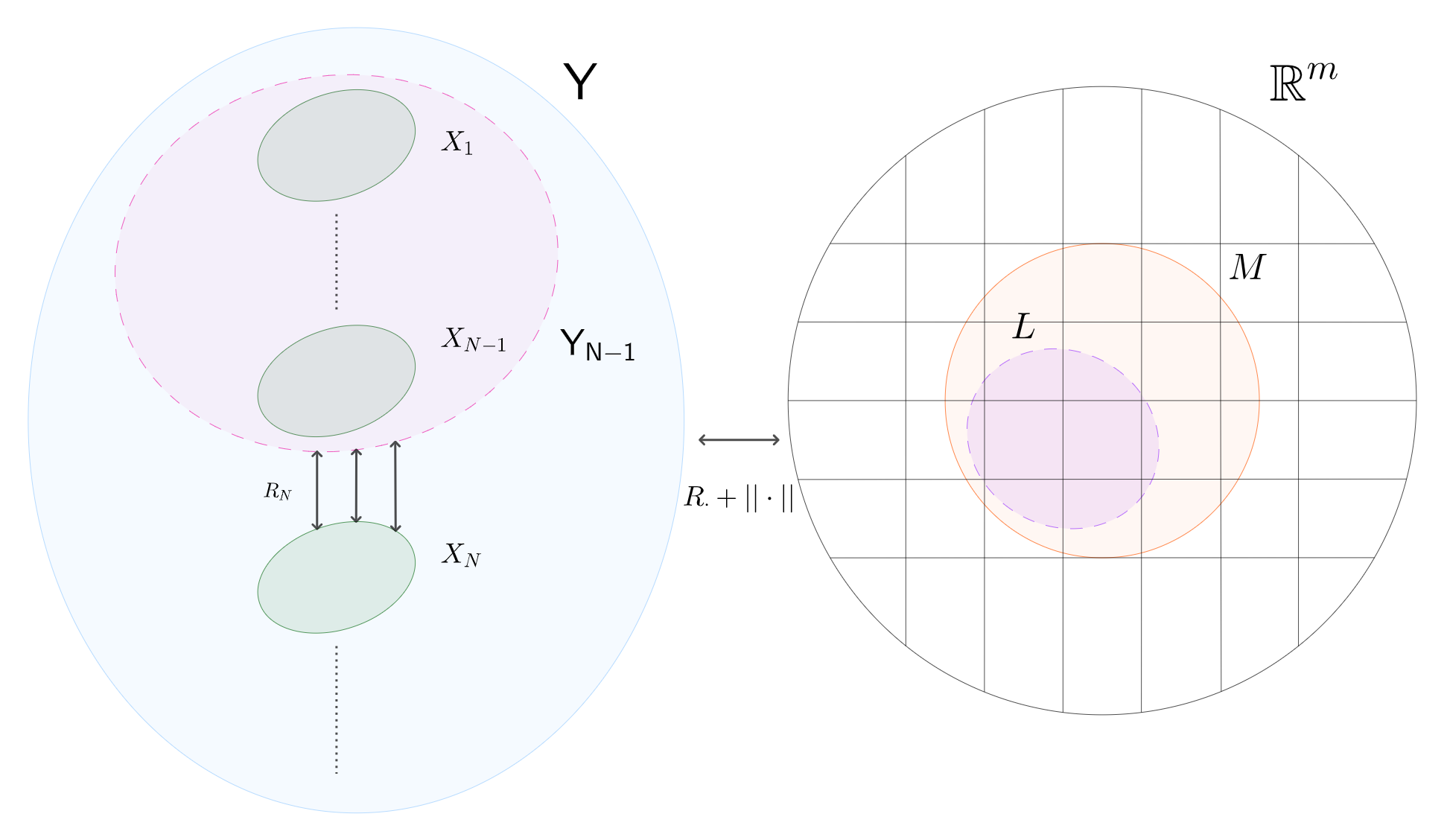}
    \caption{$(X,d)$ with $\d_\A(x) = \rho$ and $\hat{R}(X) = \infty$.}
\end{figure}

\paragraph{Phase II.}

Let us fix some $m \in \N$ such that $m > \rho \geq c$. We construct the metric space $(Z, d)$ by gluing $(Y, D)$ and some arbitrary normed space $(\R^m, \n{\cdot})$ together in the following way:
\begin{itemize}
    \item The distances between pairs of points both inside either $Y$ or $\R^m$ remain unchanged.
    \item $\forall N \in \N, \en \forall x \in X_N, \en \forall y \in \R^m$, we let $d(x, y) = R_N + \n{y}$.
\end{itemize}
We only need to check that the triangle inequality holds between the following triplets:
\begin{itemize}
    \item $M > N$, $x \in X_M$, $y \in X_N$, $z \in \R^m$: $d(x, y) = R_M$, $d(y, z) = R_N + \n{z}$, $d(z, x) = R_M + \n{z}$;
    \item $x, y \in X_N$, $z \in \R^m$: $d(x, y) < R_N$, $d(y, z) = R_N + \n{z}$, $d(z, x) = R_N + \n{z}$;
    \item $x \in X_N$, $y, z \in \R^m$: $d(x, y) = R_N + \n{y}$, $d(y, z) = \n{y-z}$, $d(z, x) = R_N + \n{z}$.
\end{itemize}
Hence, $(Z, d)$ is a well-defined metric space.

Now, let $M = B(0, r) \subset Z$ be the open ball of radius $r > 0$ around the origin of $\R^m$, and notice that when $r < R_1$, then $M \subset \R^m$. From the proof of Lemma \ref{assouad_vector}, we know that $\d_A(M) = m$, and from Theorem \ref{ivp}, we can find a countable subset $L$ of $M$ such that $\d_A(L) = \rho$. We can finally define the countable metric space promised in the beginning of the proof. Let $X = Y \cup L$ and let $X$ inherit the metric distance $d$ from $Z$. Since $\dim_A(L) = \rho$ and since there exists an open ball in $X$ which contains precisely $L$, we conclude that $\d_A(X) \geq \rho$.

In order to prove that $\d_A(X) \leq \rho$, let us consider an open ball $B(x, R)$ for some $x \in X$ and let $R > r> 0$. If $B(x, R) \subseteq L$, then there exists a constant $C_L$ such that $N_r(B(x, R)) \leq C_L \left( \frac{R}{r} \right)^\rho$. On the other hand, if there exists $y \in B(x, R) \cap Y$, then $B(x, R) \subseteq B(y, 2R)$. Hence, by Lemma \ref{doubling_assouad}, there exists a constant $C_Y$ such that
\begin{equation*}
    N_r(B(x, R)) \leq N_r(B(y, 2R)) \leq C_Y \left( \frac{2R}{r} \right)^c \leq 2^\rho C_Y \left( \frac{R}{r} \right)^\rho.
\end{equation*}

Therefore, $\d_A(X) = \rho$. Moreover, since $(Y,D) \subset (X, d)$, $\hat{R}(X) = \infty$.
\end{proof}

\bibliography{references}

\begin{thebibliography}{25}
\providecommand{\natexlab}[1]{#1}
\providecommand{\url}[1]{\texttt{#1}}
\expandafter\ifx\csname urlstyle\endcsname\relax
  \providecommand{\doi}[1]{doi: #1}\else
  \providecommand{\doi}{doi: \begingroup \urlstyle{rm}\Url}\fi

\bibitem[Antoniadis et~al.(2016)Antoniadis, Barcelo, Nugent, Pruhs, Schewior,
  and Scquizzato]{ABNPSS}
Antonios Antoniadis, Neal Barcelo, Michael Nugent, Kirk Pruhs, Kevin Schewior,
  and Michele Scquizzato.
\newblock Chasing convex bodies and functions.
\newblock In \emph{LATIN 2016: Theoretical Informatics: 12th Latin American
  Symposium, Ensenada, Mexico, April 11-15, 2016, Proceedings 12}, pages
  68--81. Springer, 2016.

\bibitem[Argue et~al.(2019)Argue, Bubeck, Cohen, Gupta, and Lee]{ABCGL}
CJ~Argue, S{\'e}bastien Bubeck, Michael~B Cohen, Anupam Gupta, and Yin~Tat Lee.
\newblock A nearly-linear bound for chasing nested convex bodies.
\newblock In \emph{Proceedings of the Thirtieth Annual ACM-SIAM Symposium on
  Discrete Algorithms}, pages 117--122. SIAM, 2019.

\bibitem[Argue et~al.(2021)Argue, Gupta, Tang, and Guruganesh]{AGTG}
CJ~Argue, Anupam Gupta, Ziye Tang, and Guru Guruganesh.
\newblock Chasing convex bodies with linear competitive ratio.
\newblock \emph{Journal of the ACM (JACM)}, 68\penalty0 (5):\penalty0 1--10,
  2021.

\bibitem[Assouad(1983)]{assouad}
Patrice Assouad.
\newblock Plongements lipschitziens dans $\mathbb{R}^n$.
\newblock \emph{Bulletin de la Soci{\'e}t{\'e} Math{\'e}matique de France},
  111:\penalty0 429--448, 1983.

\bibitem[Ball et~al.(1997)]{ball}
Keith Ball et~al.
\newblock An elementary introduction to modern convex geometry.
\newblock \emph{Flavors of geometry}, 31\penalty0 (1--58):\penalty0 26, 1997.

\bibitem[Bansa et~al.(2018)Bansa, B{\"o}hm, Eli{\'a}{\v{s}}, Koumoutsos, and
  Umboh]{BBEKU}
Nikhil Bansa, Martin B{\"o}hm, Marek Eli{\'a}{\v{s}}, Grigorios Koumoutsos, and
  Seeun~William Umboh.
\newblock Nested convex bodies are chaseable.
\newblock In \emph{Proceedings of the Twenty-Ninth Annual ACM-SIAM Symposium on
  Discrete Algorithms}, pages 1253--1260. SIAM, 2018.

\bibitem[Bubeck et~al.(2019)Bubeck, Lee, Li, and Sellke]{BLLS}
S{\'e}bastien Bubeck, Yin~Tat Lee, Yuanzhi Li, and Mark Sellke.
\newblock Competitively chasing convex bodies.
\newblock In \emph{Proceedings of the 51st Annual ACM SIGACT Symposium on
  Theory of Computing}, pages 861--868, 2019.

\bibitem[Bubeck et~al.(2020)Bubeck, Klartag, Lee, Li, and Sellke]{BKLLS}
S{\'e}bastien Bubeck, Bo'az Klartag, Yin~Tat Lee, Yuanzhi Li, and Mark Sellke.
\newblock Chasing nested convex bodies nearly optimally.
\newblock In \emph{Proceedings of the Fourteenth Annual ACM-SIAM Symposium on
  Discrete Algorithms}, pages 1496--1508. SIAM, 2020.

\bibitem[Bubeck et~al.(2021)Bubeck, Cohen, Lee, and Lee]{bubeck2021metrical}
S{\'e}bastien Bubeck, Michael~B Cohen, James~R Lee, and Yin~Tat Lee.
\newblock Metrical task systems on trees via mirror descent and unfair gluing.
\newblock \emph{SIAM Journal on Computing}, 50\penalty0 (3):\penalty0 909--923,
  2021.

\bibitem[Bubeck et~al.(2022)Bubeck, Coester, and Rabani]{bubeck2022randomized}
S{\'e}bastien Bubeck, Christian Coester, and Yuval Rabani.
\newblock The randomized $ k $-server conjecture is false!
\newblock \emph{arXiv preprint arXiv:2211.05753}, 2022.

\bibitem[Chen et~al.(2020)Chen, Wu, and Wu]{CWW}
Changhao Chen, Meng Wu, and Wen Wu.
\newblock Accessible values for the assouad and lower dimensions of subsets.
\newblock 2020.

\bibitem[Fraser(2022)]{fraser}
Jonathan~M Fraser.
\newblock Assouad dimension and fractal geometry: updates on open problems.
\newblock 2022.

\bibitem[Friedman and Linial(1993)]{FL}
Joel Friedman and Nathan Linial.
\newblock On convex body chasing.
\newblock \emph{Discrete \& Computational Geometry}, 9\penalty0 (3):\penalty0
  293--321, 1993.

\bibitem[Giannopoulos(1995)]{giannopoulos}
AA~Giannopoulos.
\newblock A note on the banach-mazur distance to the cube.
\newblock In \emph{Geometric Aspects of Functional Analysis: Israel Seminar
  (GAFA) 1992--94}, pages 67--73. Springer, 1995.

\bibitem[Gupta et~al.(2003)Gupta, Krauthgamer, and Lee]{GKL}
Anupam Gupta, Robert Krauthgamer, and James~R Lee.
\newblock Bounded geometries, fractals, and low-distortion embeddings.
\newblock In \emph{44th Annual IEEE Symposium on Foundations of Computer
  Science, 2003. Proceedings.}, pages 534--543. IEEE, 2003.

\bibitem[Heinonen(2001)]{heinonen}
Juha Heinonen.
\newblock \emph{Lectures on analysis on metric spaces}.
\newblock Springer Science \& Business Media, 2001.

\bibitem[John(1948)]{john}
Fritz John.
\newblock Extremum problems with inequalities as subsidiary conditions.
\newblock \emph{Interscience Publishers, Inc., New York. January}, 8:\penalty0
  187--204, 1948.

\bibitem[Karger and Ruhl(2002)]{KR}
David~R Karger and Matthias Ruhl.
\newblock Finding nearest neighbors in growth-restricted metrics.
\newblock In \emph{Proceedings of the thiry-fourth annual ACM symposium on
  Theory of computing}, pages 741--750, 2002.

\bibitem[Manasse et~al.(1988)Manasse, McGeoch, and Sleator]{MMS}
Mark Manasse, Lyle McGeoch, and Daniel Sleator.
\newblock Competitive algorithms for on-line problems.
\newblock In \emph{Proceedings of the twentieth annual ACM symposium on Theory
  of computing}, pages 322--333, 1988.

\bibitem[Matou{\v{s}}ek(2013)]{matousek}
Jir{\i} Matou{\v{s}}ek.
\newblock Lecture notes on metric embeddings.
\newblock Technical report, Technical report, ETH Z{\"u}rich, 2013.

\bibitem[Sellke(2023)]{sellke}
Mark Sellke.
\newblock Chasing convex bodies optimally.
\newblock In \emph{Geometric Aspects of Functional Analysis: Israel Seminar
  (GAFA) 2020-2022}, pages 313--335. Springer, 2023.

\bibitem[Tikhomirov(2019)]{tikhomirov}
Konstantin Tikhomirov.
\newblock On the banach--mazur distance to cross-polytope.
\newblock \emph{Advances in Mathematics}, 345:\penalty0 598--617, 2019.

\bibitem[Tomczak-Jaegermann(1989)]{tomczak-jaegermann}
Nicole Tomczak-Jaegermann.
\newblock Banach-mazur distances and finite-dimensional operator ideals.
\newblock \emph{(No Title)}, 1989.

\bibitem[Wang and Wen(2016)]{WW}
Wen Wang and Shengyou Wen.
\newblock An intermediate-value property for assouad dimension of metric space.
\newblock \emph{Topology and its Applications}, 209:\penalty0 120--133, 2016.

\bibitem[Youssef(2014)]{youssef}
Pierre Youssef.
\newblock Restricted invertibility and the banach--mazur distance to the cube.
\newblock \emph{Mathematika}, 60\penalty0 (1):\penalty0 201--218, 2014.

\end{thebibliography}

\appendix

\section{Proofs of Metric-Dimension Lemmas} \label{proofs}

We first note that the $\g-$cover dimension agrees up to a constant with the ordinary definition of dimension of a normed vector space.

\doublingVector*

\begin{proof}

We first need to define the concept of an $\eps-$net:

\begin{definition}
Given a metric space $(X,d)$, a set of points $\mathcal{N} \subseteq X$ is called an $\eps-$\textbf{net} of $X$ if $X = \bigcup_{x \in \mathcal{N}} B(x, \eps)$ and $d(x,y) \geq \eps, \en \forall x \neq y \in \mathcal{N}$. Note that $\eps-$nets can be constructed greedily over $X$ by starting with $\mathcal{N} = \es$ and repeatedly adding an arbitrary point from $X \setminus \bigcup_{x \in \mathcal{N}} B(x, \eps)$ to $\mathcal{N}$ until $X$ is exhausted.
\end{definition}

Now we can proceed with our proof. Let $r>0$ and $\eps \in (0,1)$. We will prove that every open ball of radius $r$ can be covered by at most $(1 + 2/\eps)^d$ balls of radius $\eps r$. Let $x \in \R^d$ and let $\mathcal{N}$ be an $\eps r -$net of $B(x, r)$. Then, $B(x,r) \subseteq \bigcup_{y \in \mathcal{N}} B(y, \eps r)$. Note that the balls $B(y, \eps r/2)$ are all disjoint, so
$$\text{vol} \br{\bigcup_{y \in \mathcal{N}} B(y, \eps r/2)} = |\mathcal{N}| (\eps r/2)^d \text{vol}(\B)$$
where $\B$ is the unit ball in $(\R^d, \n{\cdot})$. Now, 
$$\bigcup_{y \in \mathcal{N}} B(y, \eps r/2) \subseteq \bigcup_{y \in B(x, r)} B(y, \eps r/2) \subseteq B(x, r + \eps r/2)$$
so $|\mathcal{N}|(\eps r/2)^d\text{vol}(\B) \leq r^d(1 + \eps/2)^d\text{vol}(\B)$. Hence, $|\mathcal{N}| \leq (1 + 2/\eps)^d$. This allows us to conclude that $\lambda_\g \leq (1 + 2\g)^d$, or equivalently, $\d_\g(X) \leq d \log_\g (1 + 2\g)$.

Similarly, we have that $\text{vol}(B(x, r)) = r^d \text{vol}(\mathbb B) = \g^d \text{vol}(B(y, r/\g))$ for all $y$, so $\lambda_\g \geq \g^d$, thus proving that $\d_\g(X) \geq d$. Therefore, we conclude that $d \leq \d_\g(X) \leq d \log_\g(1 + 2\g)$.
\end{proof}

\doublingSubmetric*

\begin{proof}
Let us denote with $B_Y(\cdot, \cdot)$ and $B_X(\cdot, \cdot)$ the open balls in $Y$ and $X$, respectively. Hence, $\forall y \in Y, \en r >0$, $B_Y(y, r) \subseteq B_X(y, r)$, which can be covered by $\lambda_\g^{\ceil{\log_\g (2\g)}}$ open balls in $X$ of radius $r/(2\g)$. Thus, we have
$$B_Y(y, r) \subseteq B_X(y, r) \subseteq \bigcup_{x' \in X'} B_X(x', r/(2\g))$$
where $|X'| \leq \lambda_\g^{\ceil{\log_\g (2\g)}}$. If we define $Y'$ to to be a set of arbitrarily chosen $y' \in B_Y(y, r)$ from every nonempty $B_X(x', r/(2\g))$, then clearly, $B_X(x', r/(2\g)) \subseteq B_X(y', r/2)$. Hence, $\lambda_\g(Y) \leq \lambda_\g^{\ceil{\log_\g (2\g)}}$, so $\d_\g(Y) < \d_\g(X)(2 + \log_\g 2)$.
\end{proof}

\doublingUnion*

\begin{proof}
Let $\lambda_X$ and $\lambda_Y$ be the $\g-$cover constants of $X$ and $Y$. If an open ball $B(z,R)$ is completely contained in $X$ or $Y$, then $N_{R/\g}(B(z, R)) \leq \max(\lambda_X, \lambda_Y)$. On the other hand, if $B(z,R)$ contains points $x \in X$ and $y \in Y$, then $R \geq d(X,Y)/2$. Therefore, $R/\g \geq \max\{\text{diam}(X), \text{diam}(Y)\}$, so the whole space can be covered by $B(x,R/\g)$ and $B(y, R/\g)$. Therefore, the doubling constant of $\lambda_g(Z) \leq \max\{\lambda_X, \lambda_Y\}$. 
\end{proof}

\doublingAssouad*

\begin{proof}
    Fix some $\g > 1$ and let $(X, d)$ be a metric space. Clearly, $\dim_A(X) < \infty$ implies that for every $R > 0$ and $x \in X$,
    \begin{equation*}
        N_{R/\g}(B(x, R)) \leq C \g^{\dim_A(X)},
    \end{equation*}
    so $\d_\g(X) \leq \d_A(X) + \log_\g C < \infty$. On the other hand, if $\d_\g(X) < \infty$, then we have two cases to consider. First, let $x \in X, \en R > 0$, and let $r = R/\a$ such that $\a \in (0, \g)$. Then,
    \begin{equation*}
        N_{R/\a}(B(x, R)) \leq N_{R/\g}(B(x, R)) = \lambda_\g.
    \end{equation*}
    Second, if $\a \geq \g$, then
    \begin{equation*}
        N_{R/\a}(B(x, R)) \leq \lambda_\g^{\ceil{\log_\g \a}} < \lambda_\g \a^{\d_\g(X)}.
    \end{equation*}
    Hence, $\d_A(X) \leq \d_\g(X)$.
\end{proof}

\assouadVector*

\begin{proof}
    In the proof of Lemma \ref{doubling_vector}, we showed that for any $x \in \R^d$ and $0 < r < R$,
    \begin{equation*}
         \left( \frac{R}{r} \right)^d \leq N_r(B(x, R)) \leq \left( 1 + 2\frac{R}{r} \right)^d < 3^d \left( \frac{R}{r} \right)^d.
    \end{equation*}
\end{proof}

\end{document}